\documentclass[a4paper,UKenglish,numberwithinsect, cleveref, thm-restate]{lipics-v2021}
\usepackage{amsfonts}
\usepackage{amssymb}
\usepackage{amstext}
\usepackage{amsmath}
\usepackage{xspace}
\usepackage{graphicx}

\usepackage{graphics}
\usepackage{colordvi}
\usepackage{color}
\usepackage{hyperref}

\usepackage{bm}
\usepackage{mathtools}
\usepackage{nicematrix}
\usepackage{svg}
\usepackage[ruled,vlined]{algorithm2e}

\usepackage{thmtools}

\newcommand{\tw}{\mathsf{tw}}
\newcommand{\pw}{\mathsf{pw}}

\newcommand{\lset}{{\mathcal L}}
 
\newcommand{\cset}{{\mathcal C}}

\newcommand{\opt}{{\sf OPT}\xspace}

\title{Polynomial-time Approximation of Independent Set Parameterized by Treewidth}
\titlerunning{Polynomial-time Approximation of Independent Set Parameterized by Treewidth}

\acknowledgements{The research presented in this paper was initiated partially during the
trimester on Discrete Optimization at Hausdorff Research Institute for
Mathematics (HIM) in Bonn, Germany.}

\author{Parinya Chalermsook}{Aalto University, Finland}{parinya.chalermsook@aalto.fi}{}{Supported by European Research Council (ERC) under the European Union's Horizon 2020 research and innovation programme
(grant agreement No 759557).}
\author{Fedor Fomin}{University of Bergen, Norway}{fedor.fomin@uib.no}{}{Supported by the Research Council of
Norway via the project BWCA (grant no. 314528).}
\author{Thekla Hamm}{Utrecht University, The Netherlands}{thekla.hamm@gmail.com}{}{Supported by the Austrian Science Fund (FWF, project J4651-N)}
\author{Tuukka Korhonen}{University of Bergen, Norway}{tuukka.korhonen@uib.no}{}{Supported by the Research Council of
Norway via the project BWCA (grant no. 314528).}
\author{Jesper Nederlof}{Utrecht University, The Netherlands}{j.nederlof@uu.nl}{}{Supported by the project CRACKNP that has received funding from the European Research Council (ERC) under the European Union’s Horizon 2020 research and innovation programme (grant agreement No 853234).}
\author{Ly Orgo}{Aalto University, Finland}{ly.orgo@aalto.fi}{}{Supported by European Research Council (ERC) under the European Union's Horizon 2020 research and innovation programme
(grant agreement No 759557).}

\authorrunning{P.\ Chalermsook, F.\ Fomin, T.\ Hamm, T.\ Korhonen, J.\ Nederlof, L.\ Orgo} 

\Copyright{Parinya Chalermsook, Fedor Fomin, Thekla Hamm, Tuukka Korhonen, Jesper Nederlof, and Ly Orgo} 



\category{} 

\relatedversion{} 

\nolinenumbers 

\EventEditors{Inge Li G{\o}rtz, Martin Farach-Colton, Simon J. Puglisi, and Grzegorz Herman}
\EventNoEds{4}
\EventLongTitle{31st Annual European Symposium on Algorithms (ESA 2023)}
\EventShortTitle{ESA 2023}
\EventAcronym{ESA}
\EventYear{2023}
\EventDate{September 4--6, 2023}
\EventLocation{Amsterdam, the Netherlands}
\EventLogo{}
\SeriesVolume{274}
\ArticleNo{80}

\hideLIPIcs

\begin{document}

\maketitle

\begin{abstract}
We prove the following result about approximating the maximum independent set in a graph.
Informally, we show that any approximation algorithm with a ``non-trivial'' approximation ratio (as a function of the number of vertices of the input graph $G$) can be turned into an approximation algorithm achieving almost the same ratio, albeit as a function of the treewidth of $G$. More formally, we prove that for any function $f$, the existence of a polynomial time $(n/f(n))$-approximation algorithm yields the existence of a polynomial time $O(\tw \cdot\log{f(\tw)}/f(\tw))$-approximation algorithm, where $n$ and $\tw$ denote the number of vertices and the width of a given tree decomposition of the input graph.  
By pipelining our result with the state-of-the-art $O(n \cdot (\log \log n)^2/\log^3 n)$-approximation algorithm by Feige (2004), this implies an $O(\tw \cdot (\log \log \tw)^3/\log^3 \tw)$-approximation algorithm.
\end{abstract}

\newpage

\section{Introduction}
An independent set of a graph is a subset of pairwise non-adjacent vertices.
The Maximum Independent Set problem, which asks to find an independent set of maximum cardinality of a given input graph on $n$ vertices, has been among the most fundamental optimization problems  that appeared in many research areas of computer science and has been a canonical problem of study in algorithms.

In the field of \emph{approximation algorithms}, the problem is notoriously hard: It has no $O(n/2^{\log^{3/4} n})$-approximation algorithm running in polynomial time unless $NP$ can be solved in randomized quasi-polynomial time by the work of Khot and Ponnuswami~\cite{khot2006better} (building on earlier work by among others Håstad~\cite{hastad1996clique}).
The best known polynomial time approximation algorithm is an $\tilde{O}(n/\log^3 n)$-approximation by Feige~\cite{feige2004approximating}, which is almost twenty years old; here the $\tilde{O}$-notation hides factors polynomial in $\log \log n$.

Besides measuring the approximation ratio as a function of $n$, two other directions have been suggested in the literature.
One of the directions is to measure the ratio as a function of the maximum degree $d$ of the input graph.
The first improvement over the naive greedy $(d+1)$-approximation to $o(d)$ was given by Halldorsson and Radhakrishnan~\cite{njc/HalldorssonR94} in 1994.
After this, several improvements to this approximation were made~\cite{DBLP:journals/mp/AlonK98,DBLP:journals/jgaa/Halldorsson00,DBLP:conf/soda/Halperin00}, culminating in the currently best $\tilde{O}(d/\log^{1.5} d)$-approximation by Bansal, Gupta, and Guruganesh~\cite{BansalGG18} with an almost matching lower bound of $\Omega(d/\log^2 d)$ under the Unique Games Conjecture (UGC) by Austrin, Khot, and Safra~\cite{toc/AustrinKS11}; here the $\tilde{O}$-notation hides factors polynomial in $\log \log d$.

Another direction is to measure the approximation ratio as a function of the \emph{treewidth} of the input graph.
Here, a simple greedy algorithm that is based on the fact that graph of treewidth $\tw$ are $\tw$-degenerate (see Lemma~\ref{simple_n_k_IS}) achieves an approximation ratio of $(\tw+1)$.
This was improved by Czumaj, Halld{\'o}rsson, Lingas, and Nilsson~\cite{czumaj2005approximation} in 2005, who gave a $(\tw/ \log n)$-approximation algorithm when a tree decomposition of width $\tw$ is given with the input graph.
Their algorithm is quite elegant and follows easily from the observation that one can greedily partition the vertices of the graph into sets $V_1,\ldots,V_r$ such that the treewidth of $G[V_i]$ is at most $\tw/r$.
Combined with dynamic programming for independent set on graphs of bounded treewidth, this gives a $2^{\tw/r} n^{O(1)}$ time $r$-approximation for any $r$, and therefore runs in polynomial time when we set $r=\tw/\log n$, resulting in the $(\tw/ \log n)$-approximation algorithm.

Contrary to the degree-direction of approximating independent set, 
there has been no progress in the two other directions measuring the approximation ratio as a function on the number of vertices or the treewidth since the milestone results of Feige~\cite{feige2004approximating} and Czumaj et al.~\cite{czumaj2005approximation}.
It is easy to show that one cannot improve the result of Czumaj et al.~\cite{czumaj2005approximation} to a polynomial time $(\tw/(f(\tw)\log n))$-approximation for any diverging positive function $f$, assuming the Exponential Time Hypothesis (ETH).
In particular, given an input graph $G$ on $n_0$ vertices we can create a graph $G'$ on $n=2^{n_0 / f(n_0)}$ vertices by adding $n-n_0$ vertices of degree $0$. Then $G'$ has treewidth $n_0$ and the assumed algorithm is a $n^{O(1)}=2^{o(n_0)}$-time $r$-approximation for $r=n_0/(f(n_0) \log n)=1$, which violates the lower bound that Maximum Independent Set cannot be solved exactly in $2^{o(n_0)}$ time on graphs with $n_0$ vertices, assuming ETH (see e.g. for an equivalent lower bound for Vertex Cover~\cite[Theorem 14.6]{DBLP:books/sp/CyganFKLMPPS15} ) .

This ETH lower bound naturally brings us to the question of what is the best approximation ratio in terms of treewidth only.
In this paper, we essentially resolve this question by relating the approximation ratio parameterized by treewidth tightly to the approximation ratio parameterized by $n$.

Formally, as our main result we prove the following theorem:

\begin{restatable}{theorem}{maintheoremrestate}
\label{thm: main-intro} 
Let $f: {\mathbb N} \rightarrow {\mathbb N}$ be a function such that there exists an $\frac{n}{f(n)}$-approximation algorithm for Maximum Independent Set, where $n$ is the number of vertices of the input graph\footnote{We make mild assumptions on the properties of $f$, which are detailed in \Cref{sec: prelim}. Any ``reasonable'' function $f$ satisfies these assumptions.}.
Then there exists an $O\left(\frac{\tw \cdot \log {f(\tw)}}{f(\tw)}\right)$-approximation algorithm for Maximum Independent Set, where $\tw$ is the width of a given tree decomposition of the input graph.
\end{restatable}

Let $\gamma(n)$ be the approximability function of Maximum Independent Set for $n$-vertex graph (i.e., the function for which $O(\gamma(n))$-approximation exists and $o(\gamma(n))$-approximation is hard). As mentioned before, the current state of the art has provided the lower and upper bounds $\gamma(n) = \Omega(n/2^{\log^{3/4} n})$~\cite{khot2006better,hastad1996clique} and $\gamma(n) = \tilde{O}(n/\log^3 n)$ respectively~\cite{feige2004approximating}.   
Similarly, one can consider Maximum Independent Set parameterized by $\tw$ and define $\tau(\tw)$ as the approximability function of Maximum Independent Set on the setting when a tree decomposition of width $\tw$ is given.
Our result implies that the approximability functions $\gamma$ and $\tau$ are essentially the same function, so this closes the treewidth-direction of Maximum Independent Set approximation.

We find this phenomenon rather surprising. For some other parameters, such relations do not hold, e.g., when we consider the degree parameter $d$ of the input graph, the approximability function of Maximum Independent Set is $\Omega(d/\log^2 d)$ assuming UGC~\cite{toc/AustrinKS11}, while  the $\tilde{O}(n/\log^3 n)$-approximation of Feige~\cite{feige2004approximating} exists.  

Combining Theorem~\ref{thm: main-intro} with the result of Feige~\cite{feige2004approximating}, 
we obtain the following corollary.


\begin{corollary}
\label{cor:maincor}
There exists an $O\left(\frac{\tw \cdot (\log \log \tw)^3}{\log ^3 \tw}\right)$-approximation algorithm for Maximum Independent Set, where $\tw$ is the width of a given tree decomposition of the input graph.
\end{corollary}

This improves over the result of Czumaj et al.~\cite{czumaj2005approximation} 
when $\log^{1/3} n =  o\left(\frac{\log \tw}{\log \log \tw}\right)$, i.e., when $\tw$ is larger than $\exp(\tilde{\Omega}(\log^{1/3} n))$. 
It is better than the algorithm of Feige~\cite{feige2004approximating} whenever $\tw = o(n/\log \log n)$, so overall it improves the state-of-the-art in the range of parameters
\[\exp(\tilde{\Omega}(\log^{1/3} n)) \le \tw \le o(n/\log \log n).\]

These results assume that the tree decomposition is given as part of the input. 
To remove this assumption, we can use the algorithm of  Feige et al.~\cite{DBLP:journals/siamcomp/FeigeHL08} to $O(\sqrt{\log \tw})$-approximate treewidth.
In particular, their algorithm combined with \Cref{cor:maincor} yields the following corollary in the setting when a tree decomposition is not assumed as a part of the input.

\begin{corollary}
\label{cor:notd}
There exists an $O\left(\frac{\tw \cdot (\log \log \tw)^3}{\log^{2.5} \tw}\right)$-approximation algorithm for Maximum Independent Set, where $\tw$ is the treewidth of the input graph.
\end{corollary}

\paragraph*{Techniques}
On a high level, our technique behind Theorem~\ref{thm: main-intro} is as follows: First we delete a set of vertices of size at most $\opt/2$ from the graph so that each of the remaining components can be partitioned into subinstances with pathwidth at most $\tw$ and subinstances with tree decompositions of width $O(\tw)$ and depth $O(\log f(\tw))$. For the subinstances of small pathwidth, we partition the vertices into  $O(\log f(\tw))$ levels based on in how many bags of the path decomposition they occur.
Similarly, for the subinstances with $O(\log f(\tw))$-depth tree decompositions, we partition the vertices in levels based on the depth of the highest bag of the tree decomposition they occur in.
In both subinstances we argue that all vertices of all but one level can be removed, in order to make the vertices in the remaining level behave well in the decomposition, after which the remaining level can be chopped into components of size roughly $O(\tw)$ such that the size of  maximum independent again does not decrease significantly.

Although some aspects of our approach are natural,
we are not aware of arguments modifying the tree decomposition as we did here in the previous literature; we expect these arguments may have more applications for designing approximation algorithms for other $NP$-hard problem parameterized by treewidth similar to Theorem~\ref{thm: main-intro}. 

\paragraph*{Organization}
The paper is organized as follows. We give preliminaries in Section~\ref{sec: prelim}.
A major ingredient of \Cref{thm: main-intro} will be an approximation algorithm for Maximum Independent Set parameterized by pathwidth, which we will be presented in Section~\ref{sec: pw}.
Then, the approximation algorithm for Maximum Independent Set parameterized by treewidth will be presented in Section~\ref{sec: tw}.
This will use the pathwidth case as a black box.
We then conclude and present open problems in \Cref{sec:conc}.

\section{Preliminaries}
\label{sec: prelim} 


\paragraph*{Basic notation}
We refer to \cite{diestel-book} for standard graph terminology.
We use the standard notation -- \(\alpha(G)\) -- to denote the independence number, i.e., the size of a maximum independendent set, of graph \(G\).
Throughout, for a natural number \(i\) we denote the set \(\{1, \dotsc, i\}\) by \([i]\), and for two natural numbers $i \le j$ we denote the set $\{i, i+1, \ldots, j\}$ by $[i,j]$.
We use \(\log\) to denote the base-\(2\) logarithm.

\paragraph*{Tree decompositions} Given a graph $G$, a tree decomposition of $G$ consists of a tree $T$, where each node $t \in V(T)$ is associated with a subset $B_t \subseteq V(G)$ of vertices called a bag, such that 

\begin{enumerate}
        \item $\bigcup_{t \in V(T)} B_{t} = V(G)$
        \item For every edge $uv \in E(G)$, there must be some node $t$ such that $\{u, v\} \subseteq B_{t}$.
        
        \item For every vertex $v \in V(G)$, the bags $\{t: v \in B_t\}$
        are connected in $T$. 
    \end{enumerate}

The \textbf{ width} of a tree decomposition is $\max_{t \in V(T)} |B_t| - 1$.
The \textbf{treewidth} of $G$ (denoted by $\tw(G)$) is the minimum number $k$, such that $G$ has a tree decomposition of width $k$. 
When the input graph is clear from the context, we simply write $\tw$ to denote the treewidth of $G$. 

A rooted tree decomposition is a tree decomposition where one node is assigned to be the root of the tree $T$.
We use standard rooted-tree definitions when talking about rooted tree decomposition.
The depth of a rooted tree decomposition is the depth of the tree $T$, i.e., the length of the longest root-leaf path.

A rooted tree decomposition $T$ is called \textbf{nice} if it satisfies that
\begin{itemize}
        \item Every node of $T$ has at most $2$ children.
        
        \item If a node $t$ has two children $t'$ and $t''$, then $t$ is called a join node and $B_t = B_{t'} = B_{t''}$.
        
        \item If a node $t$ has one child $t'$, then either:
        \begin{enumerate}
            \item $B_{t} \subset B_{t'}$ and $|B_{t'}| = |B_t| + 1$, in which case $t$ is a forget node, or
            
            \item $B_{t'} \subset B_{t}$ and $|B_t| = |B_{t'}| + 1$, in which case $t$ is an introduce node.
        \end{enumerate}
        \item If a node \(t\) has no children we call it a leaf node.
    \end{itemize}


It is well-known that any tree decomposition can be turned into a nice tree decomposition.

\begin{lemma}[\cite{kloks1994treewidth}]
\label{thm:tree decomp computation}
For every graph $G$ on $n$ vertices, given a tree decomposition $T'$ of width $\omega$, there is a nice tree decomposition $T$ with at most $4 \cdot n$ nodes and width $\omega$ that can be computed in polynomial time.
\end{lemma}

It is possible to also assume the following additional property without loss of generality.

\begin{lemma}
\label{prop: leaf unique}
Given a tree decomposition $T'$ of width $\omega$, there exists a nice tree decomposition of width $\omega$ and at most $4n$ nodes, that can be computed in polynomial time, such that for each leaf node $t \in V(T)$, there exists a vertex $v \in B_t$ that appears in exactly one bag, i.e., the bag $B_t$ itself. 
\end{lemma} 
\begin{proof}
We use~\Cref{thm:tree decomp computation} to compute a nice tree decomposition $T$. If there exists a leaf node $t \in V(T)$ that does not contain such a vertex, we delete $t$ from $T$.
Notice that all the properties of a tree decomposition continue to hold after such a deletion. 
However, if after this deletion the former parent \(s\) of \(t\) in \(T\) is not a leaf, \(s\) was a join node which now has a child with the same bag as \(s\) which violates niceness.
To repair this we can simply contract the edge between \(s\) and its remaining child in \(T\).
It is straightforward to verify that after this \(T\) remains nice.
We can iterate the above, strictly decreasing the number of nodes of \(T\), until \(T\) has the desired property.
\end{proof}


We will use the following well-known lemma of Bodlaender and Hagerup~\cite{DBLP:journals/siamcomp/BodlaenderH98} to turn a tree decomposition into a logarithmic-depth tree decomposition, while increasing the width only by a factor of three.

\begin{lemma}[{\cite[Lemma~2.2]{DBLP:journals/siamcomp/BodlaenderH98}}]
\label{thm:logdepthtrans}
Given a tree decomposition of a graph $G$ 
of width $\omega$ and having \(\gamma\) nodes, we can compute in polynomial time a rooted tree decomposition of $G$ 
of depth $O(\log \gamma)$ and width at most $3\omega + 2$.
\end{lemma}

\paragraph*{Path decompositions} 
A path decomposition is a tree decomposition where the tree $T$ is a path.
The \textbf{pathwidth} of $G$ is the minimum number $k$, such that $G$ has a path decomposition of width $k$. It is denoted by $\pw(G)$.
A nice path decomposition is a nice tree decomposition where $T$ is a path, and the root is assigned to a degree-1 node, i.e., at one end of the path.
Note that there are no join-nodes in a nice path decomposition.

We observe that any path decomposition can be turned into a nice path decomposition with $2n$ nodes.
\begin{lemma}
\label{lem:nicepathdecomp}
For every graph $G$ on $n$ vertices, given a path decomposition $P'$ of width $\omega$, there is a nice path decomposition $P$ with $2n$ nodes and width $\omega$, that can be computed in polynomial time.
\end{lemma}
\begin{proof}
By introducing vertices one at a time and forgetting vertices one at a time we obtain a nice path decomposition where the bag of the first node is empty, the bag of the last node is empty, and on each edge exactly one vertex is either introduced or forgotten, and therefore the path decomposition has exactly $2n$ edges and $2n+1$ nodes.
We can remove the first bag that is empty to get a path decomposition with exactly $2n$ nodes.
\end{proof}

\paragraph*{Maximum independent set approximation}
Given a function \(r\) that maps graphs to numbers greater than \(1\), an \(r\)-approximation algorithm for Maximum Independent Set takes as input a graph \(G\) and outputs in polynomial time an independent set in \(G\) of size at least \(\frac{\alpha(G)}{r(G)}\).
We usually denote any occurrence of \(|V(G)|\) in \(r\) by \(n\).

Let us now detail our assumptions on the function $f$ in  \Cref{thm: main-intro}.
We assume that the approximation ratio $n/f(n)$ of the given approximation algorithm is a non-decreasing function on $n$.
This assumption is reasonable because if $n/f(n)$ would be decreasing at some point, we could improve the approximation ratio by adding universal vertices to the graph; note that adding universal vertices does not change the optimal solution, but increases $n$.
This also implies that the function $f(n)$ grows at most linearly in $n$.
We assume that for arbitrary fixed constant $c \ge 1$, it holds that $f(c \cdot n) \in O(f(n))$.
We also assume that the function $f$ can not decrease too much when $n$ grows, in particular, we assume that for arbitrary fixed constant $c \ge 1$ it holds that $f(c \cdot n) \in \Omega(f(n))$.

Moreover, we will use a basic result about finding independent sets whose size depends on the treewidth of the graph.
Recall that a graph $G$ is $d$-degenerate if there is always a vertex of degree at most $d$ in any induced subgraph of $G$. 
It is known that every graph $G$ is $\tw(G)$-degenerate:
Simply consider the vertex that is contained at a leaf bag and no other bag of a tree decomposition $T$ of any induced subgraph of \(G\) as given by \Cref{prop: leaf unique}. This vertex has degree at most $\tw(G)$.
Therefore, we obtain a following trivial algorithm for approximating Maximum Independent Set parameterized by treewidth.

\begin{lemma}
\label{simple_n_k_IS}
There is a polynomial time algorithm that given a graph $G$ on $n$ vertices finds an independent set of size at least $n/(\tw(G)+1)$.
\end{lemma} 
\begin{proof}
Iteratively assign a vertex of minimum degree to the independent set and delete its neighbors.
By the aforementioned degeneracy argument, at each iteration at most $\tw(G)+1$ vertices are deleted, so the number of iterations and the size of the found independent set is at least $n/(\tw(G)+1)$.
\end{proof}

Note that the algorithm of \Cref{simple_n_k_IS} does not need a tree decomposition as an input.

\section{Approximation parameterized by pathwidth}
\label{sec: pw} 
In this section, we prove a version of \Cref{thm: main-intro} where instead of a tree decomposition, the input graph is given together with a path decomposition.
This will be an important ingredient for proving \Cref{thm: main-intro}.
In particular, this section is devoted to the proof of the following lemma.

\begin{lemma}
\label{lem:pathwidth-main}
Let $f: {\mathbb N} \rightarrow {\mathbb N}$ be a function such that there exists an $\frac{n}{f(n)}$-approximation algorithm for Maximum Independent Set, where $n$ is the number of vertices of the input graph, and $f$ satisfies the assumptions outlined in \Cref{sec: prelim}.
Then there exists an $O\left(\frac{\pw \cdot \log {f(\pw)}}{f(\pw)}\right)$-approximation algorithm for Maximum Independent Set, where $\pw$ is the width of a given path decomposition of the input graph.
\end{lemma}

Throughout this section we will use $G$ to denote the input graph and $\pw$ to denote the width of the given path decomposition of $G$. We denote by $k = \pw+1$ the maximum size of a bag in the given decomposition.
Note that by our assumptions on the function $f$, it holds that $f(k) = \Theta(f(\pw))$.

Let us denote $\opt = \alpha(G)$.
If $\opt < \frac{n}{f(k)}$, then \Cref{simple_n_k_IS} gives us a solution of size at least
\[\frac{n}{\tw(G)+1} \ge \opt \cdot \frac{f(k)}{\tw(G)+1},\]
i.e., an $O(\tw(G)/f(k))$-approximation, which would give the desired result by the facts that $\tw(G) \le \pw$ and $f(k) = \Omega(f(\pw))$.
Therefore, in the rest of this section we will assume that $\opt \geq \frac{n}{f(k)}$.

Let $P$ be the given path decomposition of $G$.
By \Cref{lem:nicepathdecomp}, we can assume without loss of generality that $P$ is a nice path decomposition and has exactly $2 n$ bags, which we will denote by $B_1, \ldots, B_{2n}$ in the order they occur in the path.
For each $v \in V(G)$, we define the \textbf{length} of $v$ to be the number of bags in $P$ that contain $v$, and denote the length of $v$ by $\ell(v)$.
In particular,
\[\ell(v) = |\{i \in [1,2n] \colon v \in B_i\}|.\]

Then, we partition $V(G)$ into $2 + \lceil \log f(k) \rceil$ sets based on the lengths of the vertices: 
    \begin{align*}
        V_0 &= \{v \colon \ell(v) < 2k\} \\
        V_i &= \{v \colon \ell(v) \in [k \cdot 2^i, k  \cdot2^{i+1})\}, & & 1 \leq i \leq \lceil \log f(k) \rceil \\
        V' &= \{v \colon \ell(v) \geq 4k \cdot 2^{\lceil \log f(k) \rceil}\} 
    \end{align*}

Note that $(V_0, V_1, \ldots, V_{\lceil \log f(k) \rceil}, V')$ is indeed a partition of $V(G)$.
We first show that the set $V'$, which consists of the longest vertices, can only contribute to at most half of the optimal solution.

\begin{lemma}
\label{lem:smallisinlong}
It holds that $|V'| \leq \opt/2$.     
\end{lemma}
\begin{proof}
First, notice that $\sum_{v \in V(G)} \ell(v) \leq 2 n k$.
This is because $P$ has $2n$ bags, each vertex appears in $\ell(v)$ bags of $P$, and each bag of $P$ can have at most $k$ vertices appearing in it. 
Now, because for vertices $v \in V'$ we have $\ell(v) \geq 4k \cdot 2^{\lceil \log f(k) \rceil} \ge 4k \cdot f(k)$ the vertices in $V'$ contribute at least $\sum_{v \in V'} \ell(v) \ge 4k \cdot f(k) \cdot |V'|$ to the sum.
Therefore, it holds that 
\[|V'| \le \frac{2nk}{4k \cdot f(k)} \le \frac{n}{2 \cdot f(k)} \le \opt/2,\]
as desired.
\end{proof}

\Cref{lem:smallisinlong} implies that at least half of any maximum independent set in $G$ must be in the subgraph $G[V_0 \cup V_1 \cup \ldots \cup V_{\lceil \log f(k) \rceil}]$.
In the rest of this section, we will focus on the following lemma.

\begin{lemma}
\label{lem: pw apx inside V_i} 
For each $i \in [0,\lceil \log f(k) \rceil]$, there is a $O(k/f(k))$-approximation algorithm for Maximum Independent Set in $G[V_i]$. 
\end{lemma}

It is easy to see how \Cref{lem: pw apx inside V_i} implies \Cref{lem:pathwidth-main}.
For each such $G[V_i]$, we invoke Lemma~\ref{lem: pw apx inside V_i} to obtain a $O(k/f(k))$-approximate solution $S_i \subseteq V_i$. Our algorithm returns the set $S_i$ with the largest cardinality.  
Since there are at most $O(\log f(k))$ such sets, by Lemma~\ref{lem:smallisinlong} there must be some integer $i^*$ for which $\alpha(G[V_{i^*}]) \geq \Omega(\opt/\log f(k))$. 
Therefore, the returned set must have size at least
\[\frac{\Omega(\opt/\log f(k))}{O(k/f(k))} = \opt \cdot \Omega \left(\frac{f(k)}{k \cdot \log f(k)}\right) = \opt \cdot \Omega \left(\frac{f(\pw)}{\pw \cdot \log f(\pw)}\right).\]

Therefore, to finish the proof of \Cref{lem:pathwidth-main}, it remains to prove~\Cref{lem: pw apx inside V_i}.
\begin{proof}[Proof of \Cref{lem: pw apx inside V_i}.]
Recall that the bags of $P$ are denoted by $B_1, B_2,\ldots, B_{2n}$ where $B_h$ is the $h$-th bag in the order from left to right.
Let $L = \max_{v \in V_i} \ell(v)$ denote the maximum length of a vertex $v \in V_i$.
Recall that by our definition of $V_i$, it holds that if $i = 0$, then $L < 2k$, and if $i > 0$, then all vertices in $V_i$ have length between $L/2$ and $L$.
We partition the set $V_i$ into sets $X_r$ and $Y_r$ as follows. 

\begin{itemize}
    \item For each $r \in [1, \lfloor 2n/(2L) \rfloor]$, we define $X_r = B_{2L r} \cap V_i$. These sets contain the vertices of $V_i$ that appear in bags $B_{2 L}, B_{4L}, \ldots $ and the vertices in $B_{2L r}$ can never occur in the same bag as the vertices in $B_{2L r'}$, for any $r \neq r'$, since all vertices in $V_i$ have length at most $L$. Let $X=\bigcup_{r} X_{r}$. 

\item Denote the remaining vertices by $Y = V_i \setminus X$. We further partition $Y$ into sets $Y_r$ for $r \in [1, \lceil 2n/(2L) \rceil]$, where $Y_r$ contains the vertices $v \in Y$ that occur only in the bags $B_j$ in the interval $j \in [2L(r-1)+1, 2Lr-1]$.
\end{itemize}

It follows from definitions that $X \cup Y = V_i$. See~\Cref{fig:IS_in_V1} for an illustration.

\begin{figure}[h]
    \centering
    \includegraphics[width=0.9\textwidth]{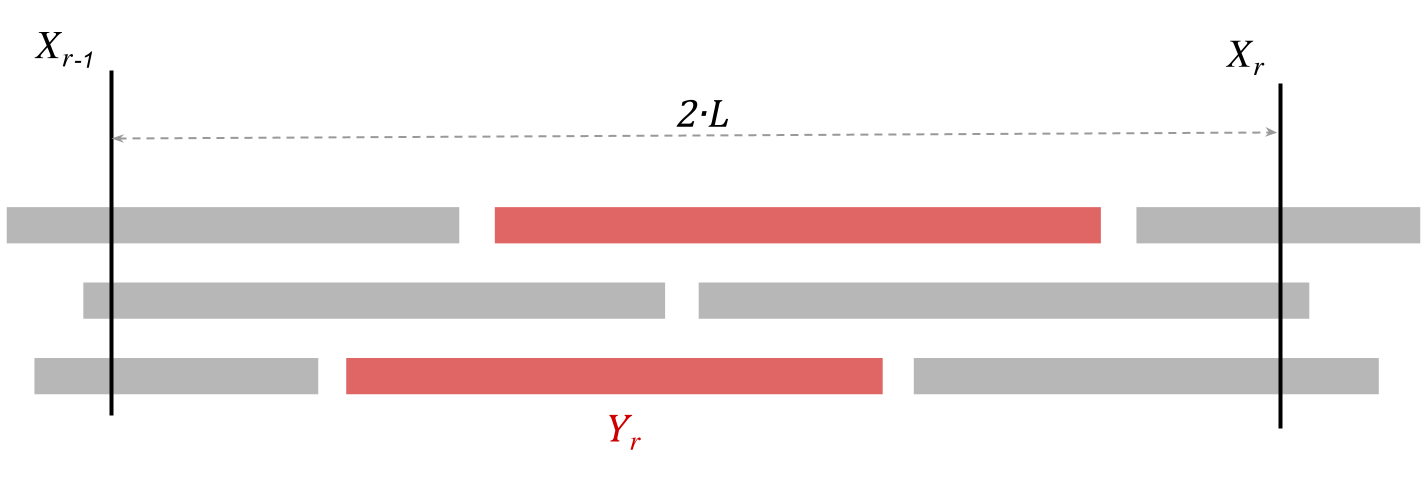}
    \caption{$X_{r}$ is the set of vertices in $V_i$ that are in the $(2 L r)$-th bag. $Y_{r}$ (in red) is the set of vertices in $V_i$ that start after $X_{r-1}$ and end before $X_{r}$. }
    \label{fig:IS_in_V1}
\end{figure}

We prove the following claim.

\begin{claim}
\label{claim:pw small set} 
For all $r \in {\mathbb N}$, both sets $X_r$ and $Y_r$ have size at most $4k$.   
\end{claim}
\begin{proof}[Proof of claim.]
For the set $X_r$, there is nothing to prove since each bag contains at most $k$ vertices. 
Let us consider the set $Y_r$.
First, we observe that because vertices of $Y_r$ occur only in the bags $B_{2L(r-1)+1}, \ldots, B_{2Lr-1}$, we have that 
\begin{equation}
\label{eq:sumlenclaim}
\sum_{v \in Y_r} \ell(v) \le k \cdot 2L,
\end{equation}
by the argument that each bag can contribute to the length of at most $k$ vertices.
Then, we consider two cases: $i>0$ and $i=0$.

In the case when $i>0$, we know that each vertex $v \in Y_r$ has length at least $\ell(v) \ge L/2$.
Together with~\cref{eq:sumlenclaim}, this implies that $|Y_r| \leq 4 k$. 

In the case when $i=0$, we have $L\leq 2k$, but we do not have the lower bound on the length of vertices in $V_0$. In this case, we use the property that $P$ is a nice path decomposition of $G$.
We know that the paths of vertices in $Y_r$ appear only in the bags $B_{2L(r-1)+1}, \ldots, B_{2Lr-1}$ of $P$.
There are $2L-1$ such bags, and because $P$ is nice, each bag either introduces a single vertex in $Y_r \cup X_r$ or forgets a single vertex in $X_{r-1} \cup Y_r$.
Since all vertices of $Y_r$ must be introduced in these bags, but there are only $2L-1$ such bags, this implies that $|Y_r| \leq 2 L-1 \le 4k$.
\end{proof}

Finally, notice that because there are no bags that contain vertices from both $Y_r$ and $Y_{r'}$ for $r \neq r'$, there are no edges between $Y_r$ and $Y_{r'}$ for $r \neq r'$.
Also, since $X_{r-1}$ and $X_{r}$ have $2L -1$ bags between them and the maximum length of a vertex is $L$, it follows that no vertex from $X_{r-1}$ occurs in a bag together with a vertex in $X_r$, and therefore there are no edges between $X_r$ and $X_{r'}$ for $r \neq r'$.
Therefore, a union of independent sets in $Y_1, \ldots, Y_{\lceil 2n/(2L) \rceil}$ is an independent set in $Y$, and a union of independent sets in $X_1, \ldots, X_{\lfloor 2n/(2L) \rfloor}$ is an independent set in $X$.

As each graph $G[X_r]$ and $G[Y_r]$ has at most $4k$ vertices, we use the given $\frac{n}{f(n)}$-approx\-imation algorithm to $\frac{4k}{f(4k)}$-approximate maximum independent set in all of the graphs $G[X_r]$ and $G[Y_r]$.
We denote by $X^*$ the union of the results in the graphs $G[X_r]$ and by $Y^*$ the union of the results in the graphs $G[Y_r]$.
Note that by previous arguments, $\alpha(G[X]) = \sum_{r} \alpha(G[X_r])$ and $\alpha(G[Y]) = \sum_{r} \alpha(G[Y_r])$, and therefore $X^*$ is a $\frac{4k}{f(4k)}$-approximation for independent set in $G[X]$ and $Y^*$ is a $\frac{4k}{f(4k)}$-approximation for independent set in $G[Y]$.
Now, we observe that because $V_i = X \cup Y$, either $\alpha(G[X]) \ge \alpha(G[V_i])/2$ or $\alpha(G[Y]) \ge \alpha(G[V_i])/2$, and therefore the larger of $X^*$ and $Y^*$ is a $\frac{8k}{f(4k)}$-approximation for independent set in $G[V_i]$.
Note that $\frac{8k}{f(4k)} = O(k/f(k))$, which is the desired approximation ratio.
\end{proof}

\section{Approximation parameterized by treewidth}
\label{sec: tw} 
In this section, we finish the proof of \Cref{thm: main-intro}.
For the convenience of the reader, let us re-state \Cref{thm: main-intro} here.

\maintheoremrestate*

Throughout this section we will use $G$ to denote the input graph and $\tw$ to denote the width of the given tree decomposition of $G$.
We denote by $k = \tw+1$ the maximum size of a bag in the given tree decomposition.
Recall that $f(k) = \Theta(f(\tw))$.

Let $T$ be the given tree decomposition of $G$.
By \Cref{prop: leaf unique} we assume that $T$ is nice, and moreover that for each leaf node $t$ of $T$ there exists a vertex $v \in B_t$ that occurs only in the bag $B_t$.
Let $\opt$ denote the size of a maximum independent set in \(G\).
Similarly to the pathwidth case in \Cref{sec: pw}, by using \Cref{simple_n_k_IS} we can assume in the rest of this section that $\opt \geq \frac{n}{f(k)}$.

Let $\lset\subseteq V(T)$ be the set of all leaf nodes of $T$. If the number of leaf nodes is at least $|\lset| \geq \frac{\opt \cdot f(k)}{k}$, then the unique vertices in these leaf bags already give us an independent set with the desired approximation factor.
Therefore, in the rest of this section we will also assume that $|\lset| < \frac{\opt \cdot f(k)}{k}$.
With this assumption, we can invoke the following lemma with \(\ell = 2 f(k)\).

\begin{lemma}
\label{lemma: removing subtrees} 
There exists a set $X \subseteq V(G)$ of size $|X| \leq k \cdot \frac{|\lset|}{\ell}$ such that for each connected component of \(G - X\) there is a rooted tree decomposition of width at most \(k-1\) that has at most \(\ell\) leaf nodes.
Such a set $X$ and the tree decompositions of the components can be computed in polynomial time. 
\end{lemma} 
\begin{proof}

We prove the lemma constructively starting with \(X = \emptyset\) and the tree decomposition \(T\) of the entire graph. We also maintain a set of tree decompositions $\cset$ of connected components of $G-X$.
We iteratively remove vertices from the graph $G$ based on the structure of $T$ as follows.

Initially, we define the set of tree decompositions we will return as $\cset = \emptyset$, and we initially assign $T'=T$ as the tree decomposition from which we will ``chop off'' pieces with at most $\ell$ leaf nodes into $\cset$.
As long as $T'$ has more than $\ell$ leaves, let $t^*$ be a node of $T'$ such that there are at least \(\ell\) leaf nodes in the subtree $T^*$ of \(T'\) rooted at \(t^*\) and no descendant of \(t^*\) has the same property. We add the vertices of $B_{t^*}$ to \(X\) and delete them from the graph and all bags of \(T'\). This separates the vertices in the bags in $T^*$ from the vertices in the bags of the rest of $T'$ and since no descendant of \(t^*\) had more than $\ell$ leaves, all connected components of $T^* - t^*$ have at most $\ell$ leaves.

We remove $T^*$ from $T'$, and add all connected components of 
$T^* - t^*$
into $\cset$. This completes the iteration. 
When the process stops, vertices in $G$ are either deleted (because they belonged to $B_{t^*}$ in some iteration) or appear in some tree decomposition that was added to $\cset$. 

By construction, each connected component of \(G - X\) has a tree decomposition that is given by a connected component in \(\cset\).
Each of these has fewer than \(\ell\) leaves and did not increase in width compared to \(T\).

With these observations, the following claim finishes the proof of the lemma.

\begin{claim}
It holds that $|X| \leq k \cdot \frac{|\lset|}{\ell}$. 
\end{claim}
\begin{proof}[Proof of claim.]
In each iteration of the algorithm, the number of leaves of $T'$ decreases by at least $\ell$, because $T^*$ has more than $\ell$ leaves.
Hence, this process terminates after at most $\frac{|\lset|}{\ell}$ iterations. Each such iteration adds a subset of a bag of \(T\) to \(X\) (which contains at most $k$ vertices). Therefore, the total number of deleted vertices is at most $k \cdot \frac{|\lset|}{\ell}$. 
\end{proof}
\end{proof}

We then assume to have \(X\) as in the statement of \Cref{lemma: removing subtrees} with \(\ell = 2 \cdot f(k)\), and for each connected component $C$ of $G-X$ a tree decomposition $T^C$ of width at most $k-1$ with at most $2 f(k)$ leaves.
For a connected component \(C\) of \(G - X\), let $S_C$ denote a fixed maximum independent set in \(C\).
Since \(|X| \leq k \cdot \frac{|\lset|}{2 f(k)} \leq \opt/2\), we know that the sum of \(|S_C|\) over all connected components \(C\) of \(G - X\) is at least \(\opt/2\).

We can distinguish two cases for a single connected component \(C\) of \(G - X\) based on whether a majority of  \(S_C\) appears in bags of nodes of degree at least 3 in \(T^C\) or not.
Formally, let $Q$ denote the set of vertices that appear in the bags of nodes of degree at least 3 in $T^C$, i.e., $Q=\bigcup_{t \text{ has degree \(>2\) in } T^C} B_t$. For each component $C$ one of the following two alternatives holds: 
\begin{enumerate}
    \item \(|S_C \setminus Q| > |S_C|/2\), or
    \item \(|S_C \cap Q | \geq |S_C|/2\).
\end{enumerate}

For handling the first case we can observe an easy pathwidth bound for \(C - Q\), which allows us to apply \Cref{lem:pathwidth-main}.
\begin{lemma}
\label{obs:casepathwidth}
A path decomposition of $C - Q$ of width at most $k-1$ can be computed in polynomial time given the tree decomposition $T^C$.
\end{lemma}
\begin{proof}
    A path decomposition witnessing this can easily be obtained from $T^C$ by deleting all nodes with degree at least 3 as well as vertices in their bags from the decomposition, resulting in a disjoint union of paths all of whose bags are of size at most \(k\).
    These paths can be concatenated in arbitrary order.
\end{proof}

For the second case we next give a lemma that splits up each \(C[Q]\) into \(O(\log f(k))\) many disjoint subgraphs in which every connected component has at most \(O(k)\) vertices.
\begin{lemma}\label{lem:casejointree}
     \(C[Q]\) can be divided into \(\ell \le O(\log {f(k)})\) subgraphs \(H_1, \dotsc, H_\ell\), such that \(V(C[Q]) = \dot{\bigcup}_{i \in [\ell]} V(H_i)\), and for any \(i \in [\ell]\), each connected component of \(H_i\) has at most \(6k\) vertices.
    Such \(H_1, \dotsc, H_\ell\) can be computed in polynomial time.
\end{lemma}
\begin{proof}

Consider the tree decomposition $T^C$ of $C$ obtained according to \Cref{lemma: removing subtrees}.
In particular, such a tree has at most $2 f(k)$ leaves and therefore at most $2 f(k)-1$ nodes with degree at least 3.

We can replace each path between two nodes \(u\) and \(v\) of degree at least 3 in \(T^C\) by two edges incident to a shared new node whose bag consists of the union $B_u \cup B_v$ of the bags of \(u\) and \(v\).
In this way we obtain a tree decomposition of \(C[Q]\) with at most \(4 \cdot f(k)\) nodes and width at most \(2k - 1\).

For this tree decomposition we invoke \Cref{thm:logdepthtrans} to obtain a tree decomposition \(T^J\) of \(C[Q]\) with width at most \(6k - 1\) and depth \(\ell \in O(\log{f(k)})\).
    Now, we partition the vertices in $C[Q]$ into $H_1,\ldots, H_{\ell}$ where $H_i$ contains all vertices $v$ such that the distance between the root of $T^J$ and  the highest bag in which $v$ appears is exactly $i-1$.

    By definition all \(V(H_i)\) are pairwise disjoint and because \(T^J\) is a tree decomposition of \(C[Q]\), the union of all \(V(H_i)\) covers \(V(C[Q])\).
    Moreover each connected component of any \(H_i\) is by construction a subset of some bag of \(T_J\) and thus has at most \(6k\) vertices as desired.
\end{proof}

With the previous lemmas in hand, we are now ready to finish the proof of \Cref{thm: main-intro} as follows.
\begin{proof}[Proof of \Cref{thm: main-intro}]
    We begin by invoking \Cref{lemma: removing subtrees}.
    Let \(\cset\) be the set of connected components in \(G - X\).

    For the next few paragraphs consider an arbitrary but fixed single connected component \(C \in \cset\).
    We first use \Cref{obs:casepathwidth} to invoke \Cref{lem:pathwidth-main} on \(C - Q\) to obtain an independent set \(S^J_C\) in \(C - Q\) of size at least \(\Omega\left(\frac{f(k)}{k \cdot \log {f(k)}}\right) \cdot \alpha(C - Q)\).

    Independently we invoke \Cref{lem:casejointree} and on each of the returned graphs \(H_i\) the assumed \(\frac{n}{f(n)}\)-approximation for \(n\)-vertex graphs on each of its connected components.
    Due to their small component size for each \(H_i\) this results in an independent set \(S_{H_i}\) of size at least \(\Omega\left(\frac{f(k)}{k}\right) \cdot \alpha(H_i)\).
    Because the graphs \(H_i\) vertex-partition \(C[Q]\) and there are only \(O(\log{f(k)})\) many \(H_i\), returning an \(S_{H_i}\) with maximum size yields an \(O(k \log f(k)/f(k) )\)-approximate solution for Maximum Independent Set on \(C[Q]\).
    
    We know that either
    \begin{enumerate}
        \item \(\alpha(C - Q) \ge \alpha(C)/2\), or
        \item \(\alpha(C[Q]) \geq \alpha(C)/2\).
    \end{enumerate}
    Overall this implies that returning the larger of \(S^J_C\) and the maximum-size \(S_{H_i}\) yields an \(O\left(\frac{k \log{f(k)}}{f(k)}\right)\)-approximate solution for Maximum Independent Set on \(C\).
    We denote the returned independent set by \(S_C\).

    Our final output is the union of all \(S_C\).
    Because all \(C\) are pairwise independent, the union of \(S_C\) is an independent set in \(G\). 
    Moreover, because $\sum_{C \in \cset}\alpha(C) \geq \opt/2$ and because of the above approximation guarantee for each \(S_C\), we obtain the overall desired approximation guarantee.
\end{proof}

\section{Conclusion and open problems}
\label{sec:conc}
In this paper we essentially settled the polynomial time approximability of Maximum Independent Set when parameterized by treewidth.
The most relevant open problem is to extend our approach to give the improved time-approximation tradeoff result in Czumaj et al.~\cite{czumaj2005approximation}. The current best known algorithm gives an $r$-approximation in $2^{\tw/r} n^{O(1)}$ time. With fine-tuning of the parameters and using the recent exponential-time approximation result of Bansal et. al.~\cite{bansal2019new}, we believe our techniques could give an improved running time of $2^{o(\tw/r)}$ when $r$ is sufficiently high, e.g., $r = \log^{\Omega(1)} \tw$. 

For us, the most interesting question is perhaps when $r$ is tiny. Can we get a $2$-approximation algorithm that runs in time $2^{(1/2-\epsilon)\tw} n^{O(1)}$? Can we prove some concrete lower bound in this regime?  While the Gap-ETH lower bound $2^{\tw/{\sf poly}(r)}$ (for sufficiently large $r$) is immediate from~\cite{bansal2019new}, such techniques do not rule out anything when $r$ is a small constant.

A different possible direction for future research would be to formulate approximation algorithms in terms of treewidth only for the more general Maximum Weight Induced Subgraph problem studied by Czumaj et al.~\cite{czumaj2005approximation}.

\bibliography{ref}

\begin{thebibliography}{10}

\bibitem{DBLP:journals/mp/AlonK98}
Noga Alon and Nabil Kahal{\'{e}}.
\newblock Approximating the independence number via the theta-function.
\newblock {\em Math. Program.}, 80:253--264, 1998.
\newblock \href {https://doi.org/10.1007/BF01581168}
  {\path{doi:10.1007/BF01581168}}.

\bibitem{toc/AustrinKS11}
Per Austrin, Subhash Khot, and Muli Safra.
\newblock Inapproximability of vertex cover and independent set in bounded
  degree graphs.
\newblock {\em Theory Comput.}, 7(1):27--43, 2011.
\newblock \href {https://doi.org/10.4086/toc.2011.v007a003}
  {\path{doi:10.4086/toc.2011.v007a003}}.

\bibitem{bansal2019new}
Nikhil Bansal, Parinya Chalermsook, Bundit Laekhanukit, Danupon Nanongkai, and
  Jesper Nederlof.
\newblock New tools and connections for exponential-time approximation.
\newblock {\em Algorithmica}, 81:3993--4009, 2019.

\bibitem{BansalGG18}
Nikhil Bansal, Anupam Gupta, and Guru Guruganesh.
\newblock On the lov{\'{a}}sz theta function for independent sets in sparse
  graphs.
\newblock {\em {SIAM} J. Comput.}, 47(3):1039--1055, 2018.
\newblock \href {https://doi.org/10.1137/15M1051002}
  {\path{doi:10.1137/15M1051002}}.

\bibitem{DBLP:journals/siamcomp/BodlaenderH98}
Hans~L. Bodlaender and Torben Hagerup.
\newblock Parallel algorithms with optimal speedup for bounded treewidth.
\newblock {\em {SIAM} J. Comput.}, 27(6):1725--1746, 1998.
\newblock \href {https://doi.org/10.1137/S0097539795289859}
  {\path{doi:10.1137/S0097539795289859}}.

\bibitem{DBLP:books/sp/CyganFKLMPPS15}
Marek Cygan, Fedor~V. Fomin, Lukasz Kowalik, Daniel Lokshtanov, D{\'{a}}niel
  Marx, Marcin Pilipczuk, Michal Pilipczuk, and Saket Saurabh.
\newblock {\em Parameterized Algorithms}.
\newblock Springer, 2015.
\newblock \href {https://doi.org/10.1007/978-3-319-21275-3}
  {\path{doi:10.1007/978-3-319-21275-3}}.

\bibitem{czumaj2005approximation}
Artur Czumaj, Magn{\'u}s~M Halld{\'o}rsson, Andrzej Lingas, and Johan Nilsson.
\newblock Approximation algorithms for optimization problems in graphs with
  superlogarithmic treewidth.
\newblock {\em Information processing letters}, 94(2):49--53, 2005.

\bibitem{diestel-book}
Reinhard Diestel.
\newblock {\em Graph Theory, 4th Edition}, volume 173 of {\em Graduate texts in
  mathematics}.
\newblock Springer, 2012.

\bibitem{feige2004approximating}
Uriel Feige.
\newblock Approximating maximum clique by removing subgraphs.
\newblock {\em SIAM Journal on Discrete Mathematics}, 18(2):219--225, 2004.

\bibitem{DBLP:journals/siamcomp/FeigeHL08}
Uriel Feige, MohammadTaghi Hajiaghayi, and James~R. Lee.
\newblock Improved approximation algorithms for minimum weight vertex
  separators.
\newblock {\em {SIAM} J. Comput.}, 38(2):629--657, 2008.
\newblock \href {https://doi.org/10.1137/05064299X}
  {\path{doi:10.1137/05064299X}}.

\bibitem{DBLP:journals/jgaa/Halldorsson00}
Magn{\'{u}}s~M. Halld{\'{o}}rsson.
\newblock Approximations of weighted independent set and hereditary subset
  problems.
\newblock {\em J. Graph Algorithms Appl.}, 4(1):1--16, 2000.
\newblock \href {https://doi.org/10.7155/jgaa.00020}
  {\path{doi:10.7155/jgaa.00020}}.

\bibitem{njc/HalldorssonR94}
Magn{\'{u}}s~M. Halld{\'{o}}rsson and Jaikumar Radhakrishnan.
\newblock Improved approximations of independent sets in bounded-degree graphs
  via subgraph removal.
\newblock {\em Nord. J. Comput.}, 1(4):475--492, 1994.

\bibitem{DBLP:conf/soda/Halperin00}
Eran Halperin.
\newblock Improved approximation algorithms for the vertex cover problem in
  graphs and hypergraphs.
\newblock In David~B. Shmoys, editor, {\em Proceedings of the Eleventh Annual
  {ACM-SIAM} Symposium on Discrete Algorithms, January 9-11, 2000, San
  Francisco, CA, {USA}}, pages 329--337. {ACM/SIAM}, 2000.
\newblock URL: \url{http://dl.acm.org/citation.cfm?id=338219.338269}.

\bibitem{hastad1996clique}
Johan Hastad.
\newblock Clique is hard to approximate within n/sup 1-/spl epsiv.
\newblock In {\em Proceedings of 37th Conference on Foundations of Computer
  Science}, pages 627--636. IEEE, 1996.

\bibitem{khot2006better}
Subhash Khot and Ashok~Kumar Ponnuswami.
\newblock Better inapproximability results for maxclique, chromatic number and
  min-3lin-deletion.
\newblock In {\em Automata, Languages and Programming: 33rd International
  Colloquium, ICALP 2006, Venice, Italy, July 10-14, 2006, Proceedings, Part I
  33}, pages 226--237. Springer, 2006.

\bibitem{kloks1994treewidth}
Ton Kloks.
\newblock {\em Treewidth: computations and approximations}.
\newblock Springer, 1994.

\end{thebibliography}

\appendix

\end{document}